\documentclass[12pt,reqno]{article}

\usepackage[usenames]{color}
\usepackage{amssymb}
\usepackage{amsmath}
\usepackage{amsthm}
\usepackage{amsfonts}
\usepackage{amscd}
\usepackage{graphicx}
\usepackage{lscape}

\usepackage[colorlinks=true,
linkcolor=webgreen,
filecolor=webbrown,
citecolor=webgreen]{hyperref}

\definecolor{webgreen}{rgb}{0,.5,0}
\definecolor{webbrown}{rgb}{.6,0,0}

\usepackage{color}
\usepackage{fullpage}
\usepackage{float}

\usepackage{graphics}
\usepackage{latexsym}
\usepackage{epsf}
\usepackage{breakurl}

\newcommand{\seqnum}[1]{\href{https://oeis.org/#1}{\rm \underline{#1}}}

\def\Enn{\mathbb{N}}

\def\Zee{\mathbb{Z}}

\newcommand{\base}[1]{{[#1]}_{-k}}
\def\andd{\, \wedge\,}

\DeclareMathOperator{\minfib}{minFib}
\DeclareMathOperator{\isfib}{isFib}
\DeclareMathOperator{\gfactoreq}{gFactorEq}
\DeclareMathOperator{\quasibifib}{quasiBiFib}

\begin{document}

\theoremstyle{plain}
\newtheorem{theorem}{Theorem}
\newtheorem{corollary}[theorem]{Corollary}
\newtheorem{lemma}[theorem]{Lemma}
\newtheorem{proposition}[theorem]{Proposition}

\theoremstyle{definition}
\newtheorem{definition}[theorem]{Definition}
\newtheorem{example}[theorem]{Example}
\newtheorem{conjecture}[theorem]{Conjecture}

\theoremstyle{remark}
\newtheorem{remark}[theorem]{Remark}

\title{Automatic Sequences in Negative Bases and\\ Proofs of Some Conjectures of Shevelev}

\author{Jeffrey Shallit,\footnote{Research funded by a grant from NSERC, 2018-04118.} \ Sonja Linghui Shan, 
and Kai Hsiang Yang\\
School of Computer Science\\
University of Waterloo\\
Waterloo, ON  N2L 3G1\\
Canada\\
\href{mailto:shallit@uwaterloo.ca}{\tt shallit@uwaterloo.ca}\\
\href{mailto:slshan@uwaterloo.ca}{\tt slshan@uwaterloo.ca}\\
\href{mailto:kh2yang@uwaterloo.ca}{\tt kh2yang@uwaterloo.ca}
}

\maketitle

\begin{abstract}
 We discuss the use of negative bases in automatic sequences.  Recently the theorem-prover {\tt Walnut} has been extended to allow the use of base $(-k)$ to express variables, thus permitting quantification over $\Zee$ instead of $\Enn$.   This enables us to prove results about two-sided (bi-infinite) automatic sequences.   We first explain the theory behind negative bases in {\tt Walnut}.   Next, we use this new version of {\tt Walnut} to give a very simple proof of a strengthened version of a theorem of Shevelev. We use our ideas to resolve two open problems of Shevelev from 2017.   We also reprove a 2000 result of Shur involving bi-infinite binary words.
\end{abstract}

\section{Introduction}

{\tt Walnut}, originally designed by Hamoon Mousavi \cite{Mousavi:2016}, is a theorem-prover that can prove or disprove first-order logical statements about automatic sequences.   Roughly speaking, automatic sequences $(a_n)_{n \geq 0}$ are those over a finite alphabet that can be computed by a DFAO (deterministic finite automaton with output) that, on input $n$ represented in some fashion, has output $a_n$.  The most famous example of such a sequence is the Thue-Morse sequence
${\bf t} = 01101001\cdots$.

As an example of the kind of statement that {\tt Walnut} can prove, consider the pattern called an {\it overlap}:  this is a word of the form $axaxa$, where $a$ is a single letter and $x$ is a possibly empty word, like the French word {\tt entente}.  The following first-order formula asserts that
${\bf t}$ has no overlaps:
\begin{equation}
 \neg \exists i,n \ (n\geq 1) \andd \forall k \ (k\leq n) \implies {\bf t}[i+k]={\bf t}[i+k+n] .
\label{one}
\end{equation}
When we enter this formula into {\tt Walnut} in a suitably-translated form, i.e.,
\begin{verbatim}
eval no_overlap "~Ei,n (n>=1) & Ak (k<=n) => T[i+k]=T[i+k+n]":
\end{verbatim}
then {\tt Walnut} returns {\tt TRUE}, thus rigorously proving the absence of overlaps in $\bf t$.   Here {\tt E} is {\tt Walnut}'s way of writing the existential quantifier, {\tt A} is the universal quantifier, {\tt =>} is logical implication, {\tt \&} is logical AND, and {\tt \char'176} is logical NOT.

The decision procedure used by
{\tt Walnut} compiles a first-order logical statement like Eq.~\eqref{one} about an automatic sequence into a series of transformations on finite automata and DFAO's (deterministic finite automata with outputs on the states).  Numbers are represented as words over a finite alphabet, in some numeration system such as base $k$.  In order for the decision procedure to work, there must be finite automata checking the addition relation $x + y = z$ and the comparison relations
$x < y$. 

Up to now, the domain of variables in {\tt Walnut} (such as the
variables $i,k,n$ appearing in Eq.~\eqref{one}), has been
restricted to $\Enn = \{0,1,2,\ldots\}$, the natural numbers.  In this paper
we describe a recent extension to {\tt Walnut} that allows
us to extend the domain of variables, in a simple and natural way, to $\Zee$, the set of all integers.   Among other
things, this enables us to mechanically prove results about
certain two-sided (bi-infinite) words (or sequences; we use these terms interchangeably).  These words can be viewed as
maps from
$\Zee$, the integers, to a finite alphabet.

We use our extension of {\tt Walnut} to improve
a 2017 result of Shevelev, and resolve two of his (up to now unproved) open problems.   We also reprove a 2000 result of 
Shur.

  For more information about {\tt Walnut}
and its capabilities, see
\cite{Shallit:2022}.

The outline of the paper is as follows.  In
Section~\ref{two} we recall the basic properties of
representation in base $-k$.  In Section~\ref{three}
we discuss automatic sequences in base $-k$.  In Section~\ref{four} we give the basic theoretical constructions that allow {\tt Walnut} to work with
base $-k$, and implementation details are given in
Section~\ref{five}.   The syntax and semantics of the new {\tt Walnut} commands are discussed in 
Section~\ref{six}; as an application we reprove a
2000 result of Shur.  In Section~\ref{seven} we give a new proof of a result of Shevelev and we also strengthen it.   In Section~\ref{eight} we solve two of Shevelev's open problems from \cite{Shevelev:2017} using our new techniques.  In Section~\ref{nine} we consider the so-called negaFibonacci representation and use it to reprove a result of Lev\'e and Richomme on the quasiperiods of the infinite Fibonacci word \cite{Leve&Richomme:2004}.   Finally, in Section~\ref{ten} we prove (and strengthen) one more result of Shevelev.

\section{Representation in base $(-k)$}
\label{two}

Let $k \geq 2$ be an integer.  In the late 19th century, Gr\"unwald \cite{Grunwald:1885} introduced the idea of
representing integers in base $(-k)$.  The history, use, and application of negative bases is described in detail in \cite[\S 4.1]{Knuth:1998} and
\cite[\S 3.7]{Allouche&Shallit:2003}.

In analogy with
ordinary base-$k$ representation, representation in base $(-k)$ involves writing
$$n = \sum_{0 \leq i \leq r} a_i (-k)^i,$$
where $n \in \Zee$ and
$a_i \in \Sigma_k$, where
$\Sigma_k := \{ 0, 1, \ldots, k-1 \}$.
Up to the inclusion of leading zeros, every integer
has a unique such base-$(-k)$ representation, called
the {\it canonical expansion}, as a word
$a_r a_{r-1} \cdots a_0$, with $a_r \not= 0$, over the alphabet $\Sigma_k$.   We denote it as $(n)_{-k}$.
Similarly, if $x \in \Sigma_k^*$ is a word,
we let $[x]_{-k}$ be the integer represented
by $x$ interpreted in base $(-k)$.

Table~\ref{tone} gives some examples of
representation in base $(-2)$---also called {\it negabinary}---where representations are given
with the most significant digit first.   Here $\epsilon$ denotes the empty word.
\begin{table}[htb]
\begin{center}
\begin{tabular}{c|c} 
$n$ & $(n)_{-2}$ \\
\hline
$-7$ &  1001  \\
$-6$ &  1110  \\
$-5$ &  1111  \\
$-4$ &  1100  \\
$-3$ &  1101  \\
$-2$ &  10    \\
$-1$ &  11    \\
 0 &  $\epsilon$      \\
 1 &  1     \\
 2 &  110   \\
 3 &  111   \\
 4 &  100   \\
 5 &  101   \\
 6 &  11010 \\
 7 &  11011 
\end{tabular}
\end{center}
\caption{Representation in base $(-2)$}
\label{tone}
\end{table}
 
The advantage to this particular representation of $\Zee$ is that we do not need artificial conventions such as a sign bit to represent integers.   It also avoids the ambiguity associated with
representations of $0$, where $-0$ and $+0$ would
have the same meaning.

\section{Automata in base $(-k)$}
\label{three}

The well-studied theory of automatic sequences
extends seamlessly to negative bases.   We say
a bi-infinite sequence $(a_n)_{n \in \Zee}$ is
{\it $(-k)$-automatic\/} if there exists a DFAO
$(Q, \Sigma_k, \Delta, \delta, q_0, \tau)$ 
where
\begin{itemize}
    \item $Q$ is a finite nonempty set of states;
    \item $\Sigma_k = \{ 0, 1, \ldots, k-1 \}$;
    \item $\Delta$ is a finite output alphabet;
    \item $\delta: Q \times \Sigma \rightarrow Q$
    is the transition function;
    \item $q_0 \in Q$ is the start state;
    \item and $\tau:Q \rightarrow \Delta$ is
    the output mapping,
\end{itemize}
such that $a(n) = \tau(\delta(q_0, x))$
for all integers $n$ and
words $x \in \Sigma^*$ such that
$[x]_{-k} = n$ (even those with leading zeroes).

Let us look at an example.   The one-sided Thue-Morse
sequence 
$$ {\bf t} = 0110100110010110 \cdots $$
can be extended in two separate ways to a two-sided generalization:
either 
$${\bf t}' = {\bf t}^R . {\bf t} = \cdots 100110010110 . 011010011001 \cdots,$$
or
$${\bf t}'' = \overline{\bf t}^R . {\bf t} = \cdots 011001101001 . 011010011001 \cdots,$$
where an $R$ as an exponent changes a one-sided right-infinite word
into a one-sided left-infinite word, the overline denotes binary
complement, and the period indicates that the $0$ index begins
immediately to the right.

Base-$(-2)$ automata for these two two-sided sequences are given
in Figure~\ref{fig1}.
\begin{figure}[htb]
\begin{center}
    \includegraphics[width=5in]{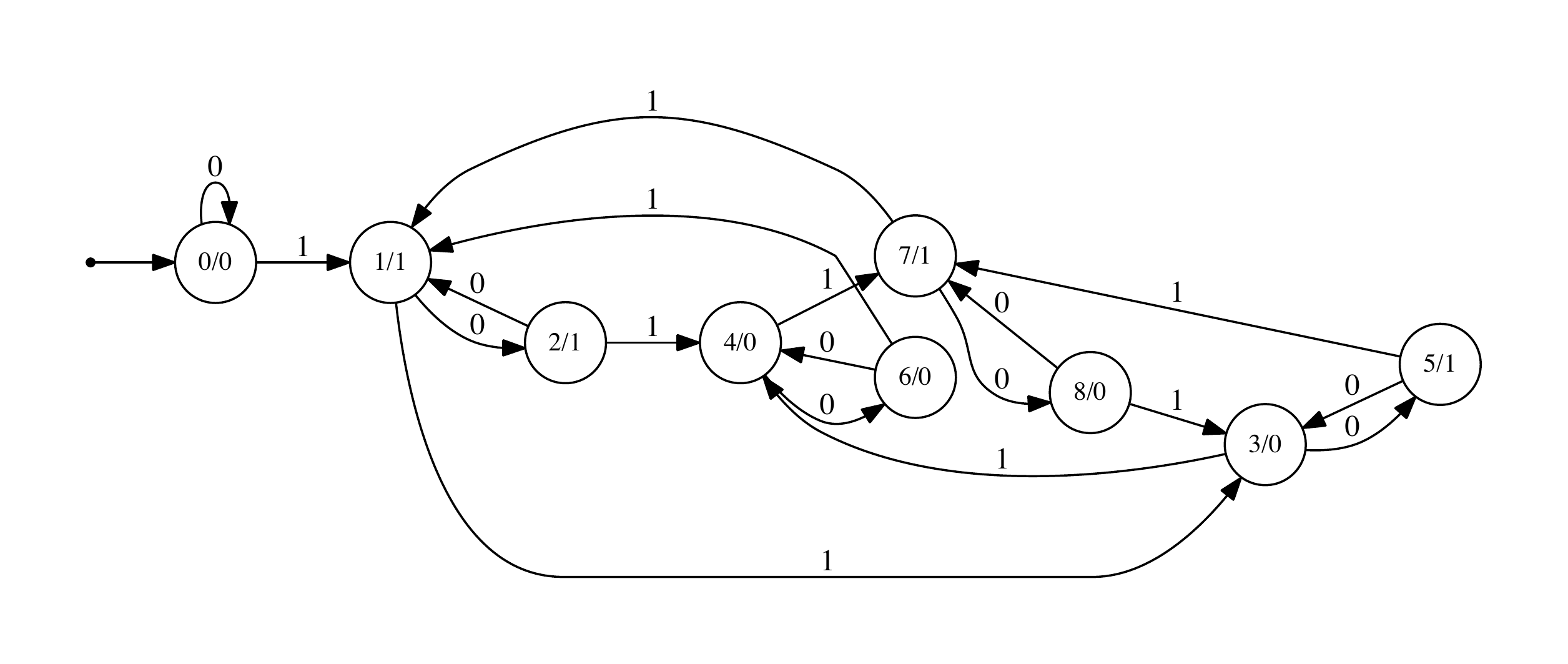} \\
    \includegraphics[width=5in]{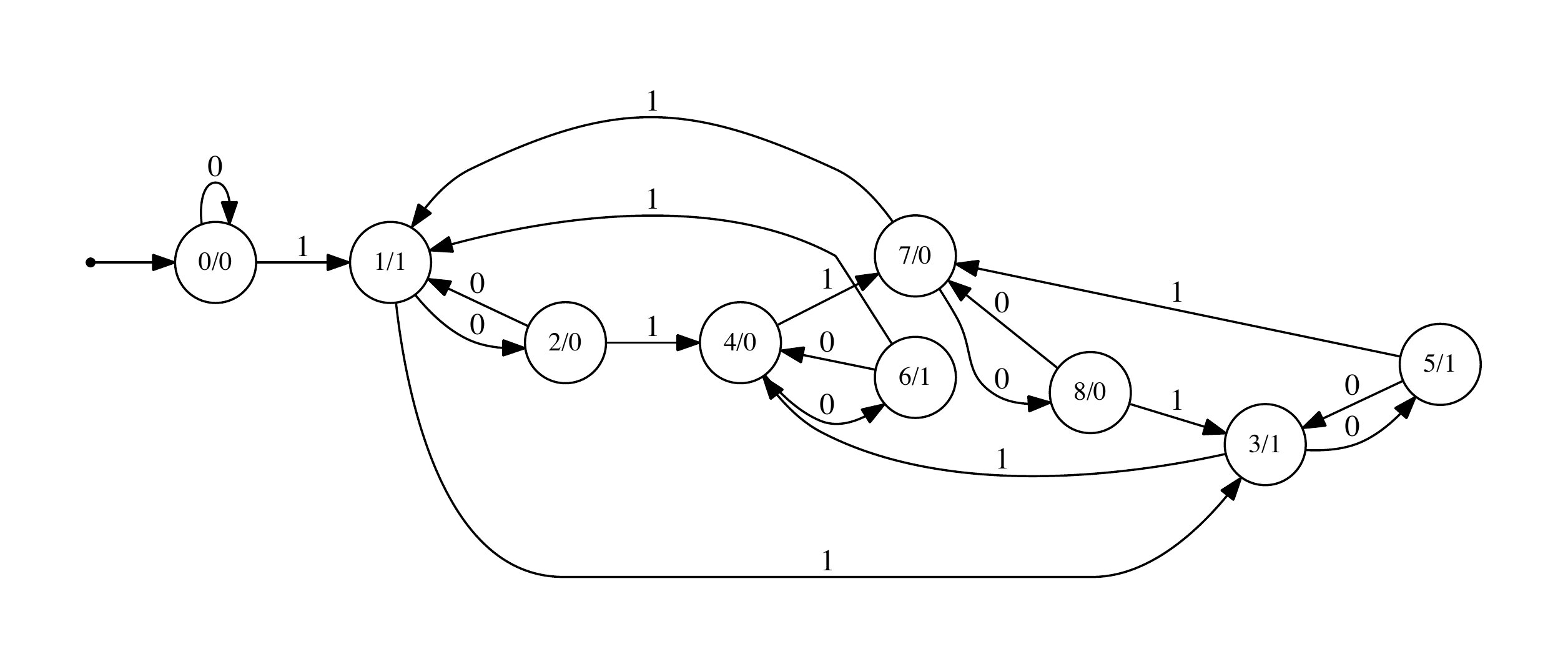}
\end{center}
\caption{Base-$(-2)$ automata generating the bi-infinite words ${\bf t}'$ (top) and ${\bf t}''$ (bottom).}
\label{fig1}
\end{figure}
Note that these automata are topologically identical; only the output
functions associated with the states differ.
We explain how to compute them in Section~\ref{twosided}.

\section{Components for working with base $(-k)$}
\label{four}

The decision procedure used by
{\tt Walnut} compiles a first-order logical statement about an automatic sequence into a series of transformations on finite automata and DFAO's (deterministic finite automata with outputs on the states).  Numbers are represented as words over a finite alphabet, in some numeration system such as base $k$.  In order for the decision procedure to work, there must be finite automata checking the addition relation $x + y = z$ and the comparison relations
$x < y$.     In this section we
describe the components needed for working with base $(-k)$.

\subsection[Adder  for base (-k)]{Adder for base \((-k)\)}\label{AutoForBaseKSection}

In this section we describe the construction of automata checking the relation $x + y = z$ in base $(-k)$.
Here $x,y,z \in \Sigma_k^*$ are all represented in base $(-k)$, and the automaton reads the representations of
$x, y, z$ in parallel, starting with the most significant digit.

Define a DFA \(M_k = (Q,\Sigma_k \times \Sigma_k \times \Sigma_k,\delta,q_0,F)\) where
\begin{enumerate}
  \item \(Q = \{ q_0, q_1, q_2, q_3 \}\);
  \item \(F = \{ q_0 \}\); and
  \item for all \([a,b,c] \in \Sigma_k \times \Sigma_k \times \Sigma_k\), we have \begin{align*}
      \delta(q_0, [a,b,c]) &= \begin{cases}
        q_0, &\text{ if } a + b - c = 0; \\
        q_1, &\text{ if } a + b - c = -1; \\
        q_2, &\text{ if } a + b - c = 1; \\
      \end{cases} \\
      \delta(q_1, [0,0,k-1]) &= q_2;\\
      \delta(q_2, [a,b,c]) &= \begin{cases}
        q_0, &\text{ if } a + b - c - k = 0;\\
        q_1, &\text{ if } a + b - c - k = -1; \\
        q_2, &\text{ if } a + b - c - k = 1,\\
      \end{cases}
  \end{align*} 
  \noindent with undefined transitions all leading to \(q_3\).
\end{enumerate}
The DFA \(M_2\) is given in \autoref{msd_neg_2_adder} with transitions to \(q_3\) hidden.
\begin{figure}[ht]
  \includegraphics[width=350pt]{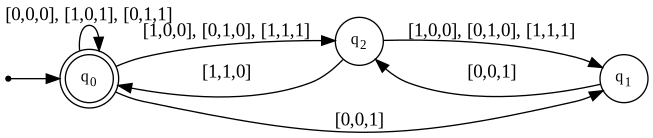}
  \centering
  \caption{\(M_2\)}\label{msd_neg_2_adder}
\end{figure}
For \(x,y,z \in \Sigma_k^*\) of length \(n\),
we interpret \([x,y,z]\) as \([x_1,y_1,z_1] [x_2,y_2,z_2] \cdots [x_n,y_n,z_n]\).
To see that \(M_k\) recognizes the language \(\{ [x,y,z] : \base{x} + \base{y} = \base{z} \}\),
we have the following result.
\begin{theorem}\label{lemma_adder}
  Let \(x,y,z \in \Sigma_k^*\) be words of length \(n\). Let \(q = \delta(q_0,[x,y,z])\).
  \begin{itemize}
    \item If \(q = q_0\) then \(\base{x} + \base{y} - \base{z} = 0\).
    \item If \(q = q_1\) then \(\base{x} + \base{y} - \base{z} = -1\).
    \item If \(q = q_2\) then \(\base{x} + \base{y} - \base{z} = 1\).
    \item If \(q = q_3\) then \(\left| \base{x} + \base{y} - \base{z} \right| > 1\).
  \end{itemize}
\end{theorem}
\begin{proof}
  We proceed by induction on \(n\).
  \begin{itemize}
    \item Suppose \(n = 0\). Then, \(q = \delta(q_0,\varepsilon) = q_0\).
      Trivially, \(\base{x} + \base{y} - \base{z} = 0\).
    \item Suppose \(n > 0\). Let \(x = x'a, y = y'b, z = z'c\) for \(x',y',z' \in \Sigma_k^*\)
      and \(a,b,c \in \Sigma_k\). Suppose the claim holds for \([x',y',z']\).
      Let \(q' = \delta(q_0,[x',y',z'])\). Then, \(q = \delta(q',[a,b,c])\).
      \begin{itemize}
        \item Case \(q' = q_0\). By the inductive hypothesis, \(\base{x'} + \base{y'} - \base{z'} = 0\). Then,
          \[\base{x} + \base{y} - \base{z} = \base{x'0} + \base{y'0} - \base{z'0} + a + b - c = a + b - c\]
          If \(q = q_0,q_1,q_2\) or \(q_3\), we have \(a + b - c = 0,-1,1\) or \(|a+b-c| > 1\) respectively.
        \item Case \(q' = q_1\). By the inductive hypothesis, \(\base{x'} + \base{y'} - \base{z'} = -1\). Then,
          \[\base{x} + \base{y} - \base{z} = \base{x'0} + \base{y'0} - \base{z'0} + a + b - c = k + a + b - c \]
          If \(q = q_2\), \(a = b = 0\) and \(c = k-1\). So \(k + a + b - c = 1\).
          If \(q = q_3\), \(a \neq 0\), \(b \neq 0\) or \(c \neq k-1\).
          Suppose \(|k + a + b - c| \leq 1\). Then, \(a + b - c \leq -k+1\).
          It follows that \(a = b = 0\) and \(c = k-1\), a contradiction.
          Hence, \(|k + a + b - c| > 1\).
        \item Case \(q' = q_2\). By the inductive hypothesis, \(\base{x'}+\base{y'} - \base{z'} = 1\).
          Then, \[\base{x} + \base{y} - \base{z} = \base{x'0} + \base{y'0} - \base{z'0} + a + b - c
          = a + b - c - k\]
          If \(q = q_0,q_1,q_2\) or \(q_3\), we have \(a + b - c - k = 0,-1,1\)
          or \(|a+b-c-k| > 1\) respectively.
        \item Case \(q' = q_3\). Since \(q' = q_3\), \(q = q_3\).
          By the inductive hypothesis, \(\left|\base{x'}+\base{y'}-\base{z'}\right| > 1\).
          It follows that \(\left|\base{x'}+\base{y'}-\base{z'}\right| \geq 2\)
          and then, \(\left|\base{x'0} + \base{y'0} - \base{z'0}\right| \geq 2k\). 
          Since \(\left|a + b - c\right| \leq 2k-2\),
          \[\left|\base{x} + \base{y} - \base{z}\right| = \left|\base{x'0} + \base{y'0} - \base{z'0}
          + a + b - c\right| > 1\]
      \end{itemize}
      In all cases, the claim holds for \([x,y,z]\).
  \end{itemize}
  This completes our proof by induction.
\end{proof}
By \autoref{lemma_adder}, we have \(M_k\) accepts \([x,y,z]\) if and only if
\(\base{x} + \base{y} = \base{z}\).

\subsection[Comparison automaton for base (-k)]{Comparison automaton for base \((-k)\)}

In this section we describe the construction of an automaton for checking if $x < y$.

Define DFA \(N_k = (Q,\Sigma_k \times \Sigma_k,\delta,q_0,F)\) where
\begin{enumerate}
  \item \(Q = \{ q_0, q_1, q_2 \}\),
  \item \(F = \{ q_2 \}\), and
  \item for all \([a,b] \in \Sigma_k \times \Sigma_k\), we have \begin{align*}
      \delta(q_0, [a,b]) &= \begin{cases}
        q_0, &\text{ if } a - b = 0; \\
        q_1, &\text{ if } a - b > 0; \\
        q_2, &\text{ if } a - b < 0; \\
      \end{cases} \\
      \delta(q_1, [a,b]) &= q_2 \\
      \delta(q_2, [a,b]) &= q_1.
  \end{align*}
\end{enumerate}
The DFA \(N_2\) is given in \autoref{msd_neg_2_less_than}.
\begin{figure}[ht]
  \includegraphics[width=250pt]{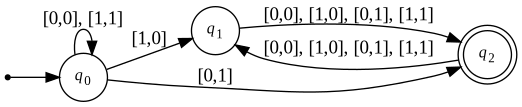}
  \centering
  \caption{\(N_2\)}\label{msd_neg_2_less_than}
\end{figure}
For \(x,y \in \Sigma_k^*\)  of length \(n\),
we similarly interpret \([x,y]\) as \([x_1,y_1] [x_2,y_2] \cdots [x_n,y_n]\).
To see that \(N_k\) recognizes the language \(\{ [x,y] : \base{x} < \base{y} \}\),
we have the following result.
\begin{theorem}\label{lemma_less_than} 
  Let \(x,y \in \Sigma_k^*\) be words of length \(n\). Let \(q = \delta(q_0,[x,y])\).
  \begin{itemize}
    \item If \(q = q_0\) then \(\base{x} - \base{y} = 0\).
    \item If \(q = q_1\) then \(\base{x} - \base{y} > 0\).
    \item If \(q = q_2\) then \(\base{x} - \base{y} < 0\).
  \end{itemize}
\end{theorem}
\begin{proof}
  We proceed by induction on \(n\).
  \begin{itemize}
    \item Suppose \(n = 0\). Then, \(q = \delta(q_0,\varepsilon) = q_0\).
      Trivially, \(\base{x} - \base{y} = 0\).
    \item Suppose \(n > 0\). Let \(x = x'a, y = y'b\) for \(x',y' \in \Sigma_k^*\) and \(a,b \in \Sigma_k\).
      Suppose the result holds for \([x',y']\). Let \(q' = \delta(q_0, [x',y'])\). Then, \(q = \delta(q',[a,b])\).
      \begin{itemize}
        \item Case \(q' = q_0\). By the inductive hypothesis, \(\base{x'} - \base{y'} = 0\). Then,
          \[\base{x} - \base{y} = \base{x'0} - \base{y'0} + a - b = a - b\]
          If \(q = q_0,q_1\) or \(q_2\), we have \(a-b = 0\), \(a-b > 0\) or \(a-b < 0\) respectively.
        \item Case \(q' = q_1\). Since \(q' = q_1\), \(q = q_2\).
          By the inductive hypothesis, \(\base{x'} - \base{y'} > 0\).
          It follows that \(\base{x'} - \base{y'} \geq 1\) and then, \(\base{x'0} - \base{y'0} \leq -k\). Then,
          \[\base{x}-\base{y} = \base{x'0}-\base{y'0} + a-b \leq a-b-k\]
          Since \(a-b < k\), \(a-b-k < 0\).
          So \(\base{x}-\base{y} \leq a-b-k < 0\).
        \item Case \(q' = q_2\). Since \(q' = q_2\), \(q = q_1\).
          By the inductive hypothesis, \(\base{x'} - \base{y'} < 0\).
          It follows that \(\base{x'} - \base{y'} \leq -1\) and then, \(\base{x'0} - \base{y'0} \geq k\). Then,
          \[\base{x}-\base{y} = \base{x'0}-\base{y'0} + a-b \geq k+a-b\]
          Since \(a-b > -k\), \(k+a-b > 0\).
          So \(\base{x}-\base{y} \geq k+a-b > 0\).
      \end{itemize}
      In all cases, the result holds for \([x,y,z]\).
  \end{itemize}
  This completes our induction proof.
\end{proof}
By \autoref{lemma_less_than}, we have \(N_k\) accepts \([x,y]\) if and only if
\(\base{x} < \base{y}\).

\subsection{Conversion from base $(-k)$ to base $k$}\label{BaseKConversionSection}

In this section we describe an automaton that can be used to convert from base $(-k)$ to base $k$ and vice versa.  
  The theory behind this can be found in
\cite[\S 5.3]{Allouche&Shallit:2003}.
\begin{figure}[H]
  \includegraphics[scale=.23]{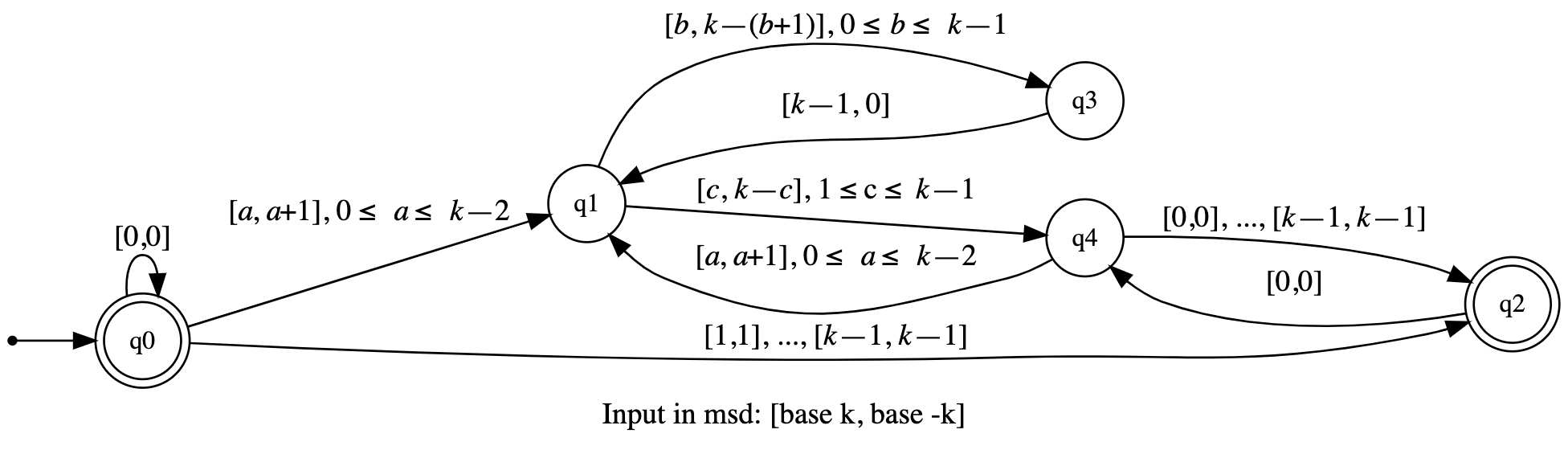}
  \centering
  \caption{$P_k$}\label{cmp_general_DFA}
\end{figure}

We define a DFA $P_k = (Q, \Sigma_k \times \Sigma_k, \delta, q_0, F)$ 
as shown in \autoref{cmp_general_DFA} where
\begin{itemize}
    \item 
    $Q = \{q_0, q_1, q_2, q_3, q_4\}$,
    \item 
    $F = \{q_0, q_2\}$, and 

    \item 
    We define the following transitions 
    \begin{align*}
        &\delta(q_0, [0, 0]) = q_0, \\
        &\delta(q_0, [a, a+1]) = q_1, 0 \leq a \leq k-2, \\
        &\delta(q_0, [m, m]) = q_2, 1 \leq m \leq k-1, \\ 
        &\delta(q_1, [b, k-(b+1)]) = q_3, 0 \leq b \leq k-1, \\ 
        &\delta(q_1, [c, k-c]) = q_4, 1 \leq c \leq k-1, \\ 
        &\delta(q_2, [0, 0]) = q_4, \\ 
        &\delta(q_3, [k-1, 0]) = q_1, \\ 
        &\delta(q_4, [a, a+1]) = q_1, 0 \leq a \leq k-2, \\
        &\delta(q_4, [n, n]) = q_2, 0 \leq n \leq k-1.
    \end{align*}
\end{itemize}

\begin{theorem}
    Let $x,y \in \Sigma_k^*$ be words of length $n$. 
    Let $q = \delta(q_0, [x,y])$.  
    \begin{itemize}
        \item 
        If $q = q_0$ then $[x]_k = [y]_{-k} = 0$.
        \item 
        If $q = q_1$ then $x, y > 0$ and $[x]_k - [y]_{-k} = -1$.
        \item 
        If $q = q_2$ then $x, y > 0$ and $[x]_k - [y]_{-k} = 0$.
        \item 
        If $q = q_3$ then $x>0$, $y<0$, and $[x]_k + [y]_{-k} = -1$.
        \item 
        If $q = q_4$ then $x>0$, $y<0$, and $[x]_k + [y]_{-k} = 0$.
    \end{itemize}
\end{theorem}
\begin{proof}
    We proceed by induction on $n$.
    
    Let $n = 0$, then $\delta(q_0, \epsilon) = q_0$ 
    and we have $[x]_k = [y]_{-k} = 0$.
    
    Let $n > 0$. Let $x = x'a$ and $y = y'b$, 
    where $x',y' \in \Sigma_k^*$ and $a, b \in \Sigma_k$.
    Let $q' = \delta(q_0, [x', y'])$ 
    and $q = \delta(q_0, [x, y]) = \delta(q', [a, b])$.
    Consider the following cases:
    \begin{itemize}
        \item 
        If $q = q_1$, then, by the construction of $P_k$,
        $n$ is odd and $x,y > 0$.
        We consider the following cases for $q'$:
        \begin{itemize}
            \item 
            If $q' = q_0$ or $q_4$ then
            $[x']_k + [y']_{-k} = 0$ by the induction hypothesis
            and $[a]_k - [b]_{-k} = -1$ by the construction of $P_k$. Then 
            \begin{align*}
                 [x]_k - [y]_{-k} 
                &= k \cdot ([x']_k + [y']_{-k}) + [a]_k - [b]_{-k} \\
                &= [a]_k - [b]_{-k} \\ 
                &= -1.
            \end{align*}
            \item 
            If $q' = q_3$ then
            $[x']_k + [y']_{-k} = -1$ by the induction hypothesis
            and $[a]_k - [b]_{-k} = k-1$ by the construction of $P_k$. Then 
            \begin{align*}
                 [x]_k - [y]_{-k} 
                &= k \cdot ([x']_k + [y']_{-k}) + [a]_k - [b]_{-k} \\
                &= -k + k -1  \\ 
                &= -1.
            \end{align*}
        \end{itemize}
        
        \item 
        If $q = q_2$, then, by the construction of $P_k$,
        $n$ is odd and $x,y > 0$.
        For both $q' = q_0$ and $q' = q_4$,
        $[x']_k + [y']_{-k} = 0$ by the induction hypothesis
        and $[a]_k - [b]_{-k} = 0$ by the construction of $P_k$. Then 
        \begin{align*}
             [x]_k - [y]_{-k} 
            &= k \cdot ([x']_k + [y']_{-k}) + [a]_k - [b]_{-k} \\
            &= 0.
        \end{align*}
        
        \item 
        If $q = q_3$, then, by the construction of $P_k$,
        $n$ is even and $x>0$ whereas $y<0$.
        Since $q' = q_1$,
        $[x']_k - [y']_{-k} = -1$ by the induction hypothesis
        and $[a]_k + [b]_{-k} = k-1$ by the construction of $P_k$. Then 
        \begin{align*}
             [x]_k + [y]_{-k} 
            &= k \cdot ([x']_k - [y']_{-k}) + [a]_k + [b]_{-k} \\
            &= -k + k - 1 \\
            &= -1.
        \end{align*}
        
        \item 
        If $q = q_4$, then, by the construction of $P_k$,
        $n$ is even and $x>0$ whereas $y<0$.
        We consider the following cases for $q'$:
        \begin{itemize}
            \item 
            If $q' = q_1$ then
            $[x']_k - [y']_{-k} = -1$ by the induction hypothesis
            and $[a]_k + [b]_{-k} = k$ by the construction of $P_k$. Then 
            \begin{align*}
                 [x]_k + [y]_{-k} 
                &= k \cdot ([x']_k - [y']_{-k}) + [a]_k + [b]_{-k} \\
                &= -k + k \\ 
                &= 0.
            \end{align*}
            \item 
            If $q' = q_2$ then
            $[x']_k - [y']_{-k} = 0$ by the induction hypothesis
            and $[a]_k + [b]_{-k} = 0$ by the construction of $P_k$. Then 
            \begin{align*}
                 [x]_k + [y]_{-k} 
                &= k \cdot ([x']_k - [y']_{-k}) + [a]_k + [b]_{-k} \\
                &= 0.
            \end{align*}
        \end{itemize}
    \end{itemize}
\end{proof}

\section[Extending Walnut to base (-k)]{Extending {\tt Walnut} to base \((-k)\)}
\label{five}

In this section, we describe how to extend {\tt Walnut} to base \((-k)\).
Suppose \(x,y,z \in \mathbb{Z}\) are written in base \((-k)\).
Recall from Section~\ref{AutoForBaseKSection}, we may check the relations \(x + y = z\) and \(x < y\)
using the base-\((-k)\) automata. 
These observations allow extending {\tt Walnut} to base \((-k)\).
Specifically, the comparison automaton \(N_k\) allows checking relations involving the relational
operators, \(=, \neq, <, >, \leq\), and \(\geq\), available in {\tt Walnut}.
Moreover, the handling of logical operators, quantification, and indexing into automata does not need to be
changed to work with specifically with base \((-k)\).
It remains to discuss how {\tt Walnut} can be extended to handle arithmetic operators in base \((-k)\).

{\tt Walnut} for base $(-k)$ is available at\\
\centerline{\url{ https://github.com/jono1202/Walnut} \ .}.

\subsection[Constants and negation in base (-k)]{Constants and negation in base \((-k)\)}

In base \((-k)\), every representation matching \({\tt 0}^* {\tt 1}\)
represents the integer \(1\). So the relation \(x = 1\) may be checked by the automaton corresponding to
\({\tt 0}^* {\tt 1}\). 
Using this fact, the relation \(x = m\) can be checked using \({\tt Walnut}\) when \(m > 1\).
Similarly, the relation \(x = 0\) may be checked by the automaton corresponding to \({\tt 0}^*\).
Alternatively, when \(m < 0\), the relation \(x = m\) can be rewritten as
\[\exists y,z \in \mathbb{Z}\  x + y = z \, \wedge\, y = (-m) \, \wedge\, z = 0. \]
This allows checking the relation \(x = m\) for all \(m \in \mathbb{Z}\).
More generally, if \(p(x_1, \dots, x_m)\) is an arithmetic expression involving variables \(x_1, \dots, x_m\),
the relation \(x = -p(x_1, \dots, x_m)\) can be rewritten as 
\[\exists y,z \in \mathbb{Z} \  x + y = z\, \wedge\, y = p(x_1, \dots, x_m) \, \wedge\, z = 0 .\]
In this way, we have the ability to negate any arithmetic expression.

\subsection[Arithmetic expressions in base (-k)]{Arithmetic expressions in base \((-k)\)}

In this section, we show how addition, subtraction, multiplication by a constant,
and division by a nonzero constant can be done in base \((-k)\).
The ability to check relations \(x = m\) for all \(m \in \mathbb{Z}\),
allows checking addition or subtraction, possibly involving integer constants.
For example, \(m - x = y\) can be rewritten as
\[\exists z\ x + y = z \wedge z = m\]
When \(m \geq 0\), the relation \(x = my\) can be handled in the same way as in non-negative bases.
When \(m < 0\), \(my\) can be rewritten as the negation of the arithmetic expression \((-m) y\).
Specifically, we rewrite \(x = my\) as \(x = -((-m) y)\).
This allows checking the relation \(x = my\) for all \(m \in \mathbb{Z}\).
Finally, we need to be able to check the relation \(x = \lfloor y/m \rfloor\) for nonzero integers \(m\).
When \(m < 0\), we rewrite \(x = \lfloor y/m \rfloor\) as
\[\exists r \in \mathbb{Z}\ y = xm + r \wedge m < r \leq 0 \]
Similarly, when \(m > 0\), we rewrite \(x = \lfloor y/m \rfloor\) as
\[\exists r \in \mathbb{Z}\ y = xm + r \wedge 0 \leq r < m \]
This allows checking the relation \(x = \lfloor y/m \rfloor\) for nonzero integer constants \(m\).

\section{New {\tt Walnut} commands}
\label{six}

In this section, we describe the new split, reverse-split, and join commands in {\tt Walnut} which will be useful for working in base \((-k)\).

{\tt Walnut} with the split, reverse-split, and join commands is available at\\
\centerline{\url{ https://github.com/jono1202/Walnut} \ .}.

\subsection[Syntax of base -k commands]{Syntax of base \((-k)\) comands}

Analogously to base \(k\), the phrase \({\tt ?msd\_neg\_k}\) preceding a formula specifies the
formula should be evaluated in base \((-k)\), with most significant digit first.
Alternatively, the phrase \({\tt ?lsd\_neg\_k}\) preceding a formula specifies the formula should be
evaluated in base \((-k)\), with least significant digit first.
The unary negation operation is written using \({\tt \_}\) (an underscore), and
can be written preceding any arithmetic expression, variable, or constant.

One thing to keep in mind when using base $(-k)$ is that when using these,
the universal and existential quantifiers now apply to all of $\Zee$, not just $\Enn$.

\subsection{The split command}

The observation from Section~\ref{BaseKConversionSection} implies that, given a \((-k)\)-automatic sequence \({(a_n)}_{n \in \mathbb{Z}}\), the corresponding transformed sequence \({(a_n)}_{n \in \mathbb{N}}\) is \(k\)-automatic.
Moreover, the base-\(k\) DFAO for \({(a_n)}_{n \in \mathbb{N}}\) can be computed using existing {\tt Walnut} functionality.

For sake of example, suppose \({(a_n)}_{n \in \mathbb{Z}}\) is over the output alphabet \(\{{\tt 0}, {\tt 1}\}\), \({\tt A}\) is the base-\((-k)\) DFAO for \({(a_n)}_{n \in \mathbb{Z}}\), and \({\tt C}\) is the conversion automaton from base-\(k\) to base-\((-k)\).
Using existing \({\tt Walnut}\) commands, we may produce the base-\(k\) DFAO for \({(a_n)}_{n \in \mathbb{Z}}\) as follows.
\begin{verbatim}
eval a "En A[n] = @1 & C(m,n)";
combine AS a;
\end{verbatim}

The new split command is used to more easily compute the same automaton, like so.
\begin{verbatim} 
split AS A[+];
\end{verbatim}
The split command can also produce the automaton for \({(a_{-n})}_{n \in \mathbb{Z}}\) and be used on DFAO with any amount of inputs.
For example, given a base-\((-k)\) DFAO, \({\tt B}\), with two inputs, we may use split to transform \({\tt B}\) like so.
\begin{verbatim}
split BS B[+] [-];
\end{verbatim}
\({\tt BS}\) outputs on \((x,y)\) in base $k$ the same output as \({\tt B}\) does on \((x,-y)\) in base \((-k)\).

\subsection{The reverse-split command}

The observation from Section~\ref{BaseKConversionSection} also implies that, given a \(k\)-automatic sequence \({(a_n)}_{n \in \mathbb{N}}\), the corresponding transformed sequence \({(a_n)}_{n \in \mathbb{Z}}\), where \(a_n := 0\) if \(n < 0\), is \((-k)\)-automatic.
The base-\((-k)\) DFAO for \({(a_n)}_{n \in \mathbb{Z}}\) can be computed using the new reverse-split command.
\begin{verbatim}
rsplit AR[+] A;
\end{verbatim}
Also similar to the split command, the reverse-split command can produce the automaton for \({(a_{-n})}_{n \in \mathbb{Z}}\) and be used on DFAO with any number of inputs.
For example, given a base-\(k\) DFAO, \({\tt B}\), with two inputs, we may use reverse-split as follows:
\begin{verbatim}
rsplit BR[+] [-] B;
\end{verbatim}
When \(x,-y \in \mathbb{N}\), \({\tt BR}\) outputs on \((x,y)\) in base \((-k)\) the same output as \({\tt B}\) does on \((x,-y)\) in base \(k\). Otherwise \({\tt BR}\) outputs zero.

\subsection{The join command}

Given a list of input DFAO, the join command produces the DFAO which outputs the first nonzero value of the input DFAO. For example, given two base-\(k\) DFAOs, \({\tt A}\) and \({\tt B}\), both with two inputs, the join command can be used like so.
\begin{verbatim}
join ABJ A[x] [y] B[x] [y];
\end{verbatim}
\({\tt ABJ}\) outputs on \((x,y)\) the first nonzero output of \({\tt A}\) or \({\tt B}\) at \((x,y)\).
If both \({\tt A}\) and \({\tt B}\) output zero on \((x,y)\), \({\tt ABJ}\) will output zero on \((x,y)\).

\subsection{Examples}
\label{twosided}

As an illustration of these new 
{\tt Walnut} commands, we show how to obtain the
automata in Figure~\ref{fig1} that extend the
one-sided Thue-Morse word ${\bf t}$ to 
to two bi-infinite sequences ${\bf t}'$ and
${\bf t}''$.
\begin{verbatim}
def tmn "T[n-1]=@1":
combine T2 tmn:

rsplit T21 [+] T:
rsplit T22 [-] T2:
join TM21 T21[x] T22[x]:

def tmn2 "T[n-1]=@0":
combine T3 tmn2:

rsplit T23 [-] T3:
join TM22 T21[x] T23[x]:
\end{verbatim}
Here we create the first automaton, {\tt TM21.txt},
by joining two automata, one for Thue-Morse on
the non-negative integers, and one for Thue-Morse
on the negative integers.  The second automaton,
{\tt TM22.txt}, is constructed similarly, except
we use $\overline{\bf t}$ for the negative integers.

To see that these two infinite words
${\bf t}'$ and ${\bf t}''$ are true generalizations of $\bf t$ to the bi-infinite case, we can prove the following theorem.
\begin{theorem}
A finite word $x$ is a factor of $\bf t$ iff it is a factor of ${\bf t}'$, iff it is a factor of ${\bf t}''$.
\end{theorem}

\begin{proof}
   We can prove this with {\tt Walnut} as follows:
\begin{verbatim}
def tm21faceq "?msd_neg_2 At (t>=0 & t<n) => TM21[i+t]=TM21[j+t]":
def tm22faceq "?msd_neg_2 At (t>=0 & t<n) => TM22[i+t]=TM22[j+t]":
eval tmtest1 "?msd_neg_2 Ai,n Ej (j>=0) & $tm21faceq(i,j,n)":
eval tmtest2 "?msd_neg_2 Ai,n Ej (j>=0) & $tm22faceq(i,j,n)":
\end{verbatim}
\noindent and {\tt Walnut} returns
{\tt TRUE} for both.
\end{proof}

We now use the automaton for ${\bf t}'$ to reprove a result of Shur \cite{Shur:2000}, namely
\begin{theorem}
There exists a bi-infinite binary word that avoids
$({7 \over 3}+\epsilon)$-powers, and has unequal frequencies of letters.
\label{shur}
\end{theorem}

Let us recall the relevant definitions.
We say a finite word $x = x[1..n]$ of length $n$ has {\it period $p$} if $x[i]=x[i+p]$ for
$1 \leq i \leq n-p$.
Let $p > q \geq 1$ be integers.  We say that a word $x$ is a {$(p/q)$-power}
if $x$ is of length $p$ and has period
$q$.  Let $\alpha > 1$ be a real number.  We say an infinite word $\bf x$ is
$\alpha$-power-free if $\bf x$ has no finite factor that is a $(p/q)$-power, for $p/q \geq \alpha$, and we say it is
$\alpha^+$-power-free if $\bf x$ has no finite factor that is a $(p/q)$-power, for $p/q > \alpha$.   For example, an overlap-free word is $2^+$-power-free.

\begin{proof}
The idea is to use the morphism given in
\cite[Theorem 33]{Shallit&Shur:2019}; it was
called $g$ there, but we call it $\xi$ here:
\begin{align*}
    \xi(0) &= 011001001101001011010011001 \\
    \xi(1) &= 011001001101001011001101001 \\
    \xi(2) &= 011001001101001100101101001
\end{align*}
First we create a bi-infinite ternary word ${\bf vtm2}$ from the
Thue-Morse word; this is a generalization of the so-called {\it variant Thue-Morse} word
{\bf vtm}.  Then we apply $\xi$ to it.
However, it is easy to see that the image of each
letter has $14$ $0$'s and $13$ $1$'s, so the
resulting word has unequal letter frequencies.
We then assert that the resulting word has a
$({7 \over 3}+\epsilon)$-power with {\tt Walnut}, as
follows:
\begin{verbatim}
morphism xi "0 -> 011001001101001011010011001 
1 -> 011001001101001011001101001 
2 -> 011001001101001100101101001":
# Shur's 27-uniform morphism

def vtm0 "?msd_neg_2 TM21[n]+1=TM21[n-1]":
def vtm1 "?msd_neg_2 TM21[n]=TM21[n-1]":
def vtm2 "?msd_neg_2 TM21[n]=TM21[n-1]+1":
combine VTM2 vtm0=0 vtm1=1 vtm2=2:

image SHUR xi VTM2:
# apply xi to VTM2 

eval testshur "?msd_neg_2 Ei,n (n>=1) & At (t>=i & 3*t<=3*i+4*n) => 
   SHUR[t]=SHUR[t+n]":
\end{verbatim}
And {\tt Walnut} returns {\tt FALSE}, so the
word is indeed $({7\over 3}+\epsilon)$-free.
\end{proof}

This was a big calculation in {\tt Walnut}, using
5468 seconds of CPU time and 400G of RAM.

\section{Shevelev's first sequence}
\label{seven}

Now that we have extended {\tt Walnut} to handle inputs
with domain $\Zee$ instead of $\Enn$, we are ready to
solve Shevelev's problems on bi-infinite sequences.  

Vladimir Shevelev  \cite{Shevelev:2017} considered the following natural analogue of the Thue-Morse sequence for base $-2$:
let $g(n)$ denote the number of $1$'s in the
base-$(-2)$ representation of $n$, taken modulo $2$.
Table~\ref{gtab} gives the first few
values of $g$.
\begin{table}[H]
\begin{center}
    \begin{tabular}{c|rrrrrrrrrrrrrrr}
    $n$ & $-7$ & $-6$ & $-5$ & $-4$ & $-3$ & $-2$ & $-1$ & 0 & 1 & 2 & 3 & 4 & 5 & 6 & 7 \\
    \hline
    $g(n)$ & 0 & 1 & 0 & 0 & 1 & 1 & 0 & 0 & 1 & 0 & 1 & 1 & 0 & 1 & 0
\end{tabular}
\end{center}
\caption{Some values of $g$}
\label{gtab}
\end{table}
The $(-2)$-automaton for this sequence is given in Figure~\ref{gaut}.   For $n \geq 0$, this is sequence 
\seqnum{A269027} in the {\it On-Line Encyclopedia of Integer Sequences} (OEIS) \cite{Sloane:2022}.
\begin{figure}[H]
\begin{center}
    \includegraphics[width=3in]{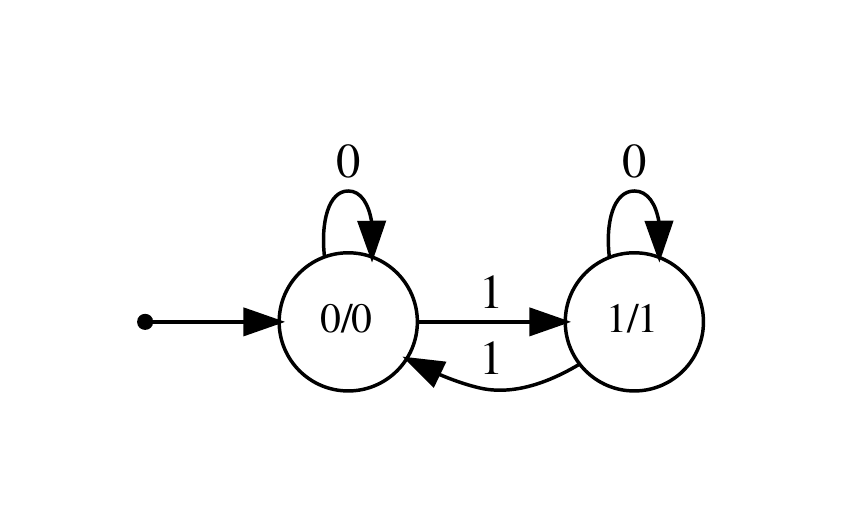}
    \end{center}
    \caption{Automaton for Shevelev's $g$ in base $-2$.}
    \label{gaut}
\end{figure}

Shevelev proved that the bi-infinite sequence
${\bf g} = (g(n))_{n \in \Zee}$ contains no cubes, that is, three consecutive identical blocks of digits.   His proof was quite complicated.

We can prove Shevelev's result and strengthen it as follows:
\begin{theorem}
The sequence $\bf g$ contains no overlaps.
\end{theorem}

\begin{proof}
We use the new version of {\tt Walnut} that can use negative bases.   We run the following command, which checks the assertion that there are overlaps.  {\tt Walnut} returns {\tt FALSE}, so there are no overlaps.
\begin{verbatim}
eval shevelev "?msd_neg_2 Ei,n (n>=1) & At (t>=0 & t<=n) => G[i+t]=G[i+t+n]":
\end{verbatim}
\end{proof}

\section{Shevelev's two open problems}
\label{eight}

In this section we show how, using
{\tt Walnut}, to prove
Shevelev's two open problems, labeled (C)
and (D) in his paper
\cite{Shevelev:2017}.  They concern
two sequences $v$ and $w$.

\subsection{The $v(n)$ sequence}
\label{vnseq}
Shevelev also studied the following sequence:
define $v(n)$ to be the largest integer $k$
such that $g(i) = t(n+i)$ for $0 \leq i < k$, where
$t(n)$ is the Thue-Morse sequence.   This sequence
is $2$-synchronized in the sense of 
\cite{Shallit:2021b}; this means there is an automaton accepting precisely the base-$2$
representations of $n$ and $v(n)$ in parallel, where we pad the shorter one with leading zeros.

We can construct this automaton with {\tt Walnut} as follows:
first, we extract the values of $g$ on $\Enn$
using the new {\tt split} command.  It produces a new $2$-automaton {\tt SH} 
for $g(n)$.  Of course {\tt SH} can only
compute $g(n)$ for non-negative values
of $n$.  We do it with the following
command:
\begin{verbatim}
split SH G [+]:
\end{verbatim}
This produces the automaton in
Figure~\ref{shev}.
\begin{figure}[H]
\begin{center}
    \includegraphics[width=6in]{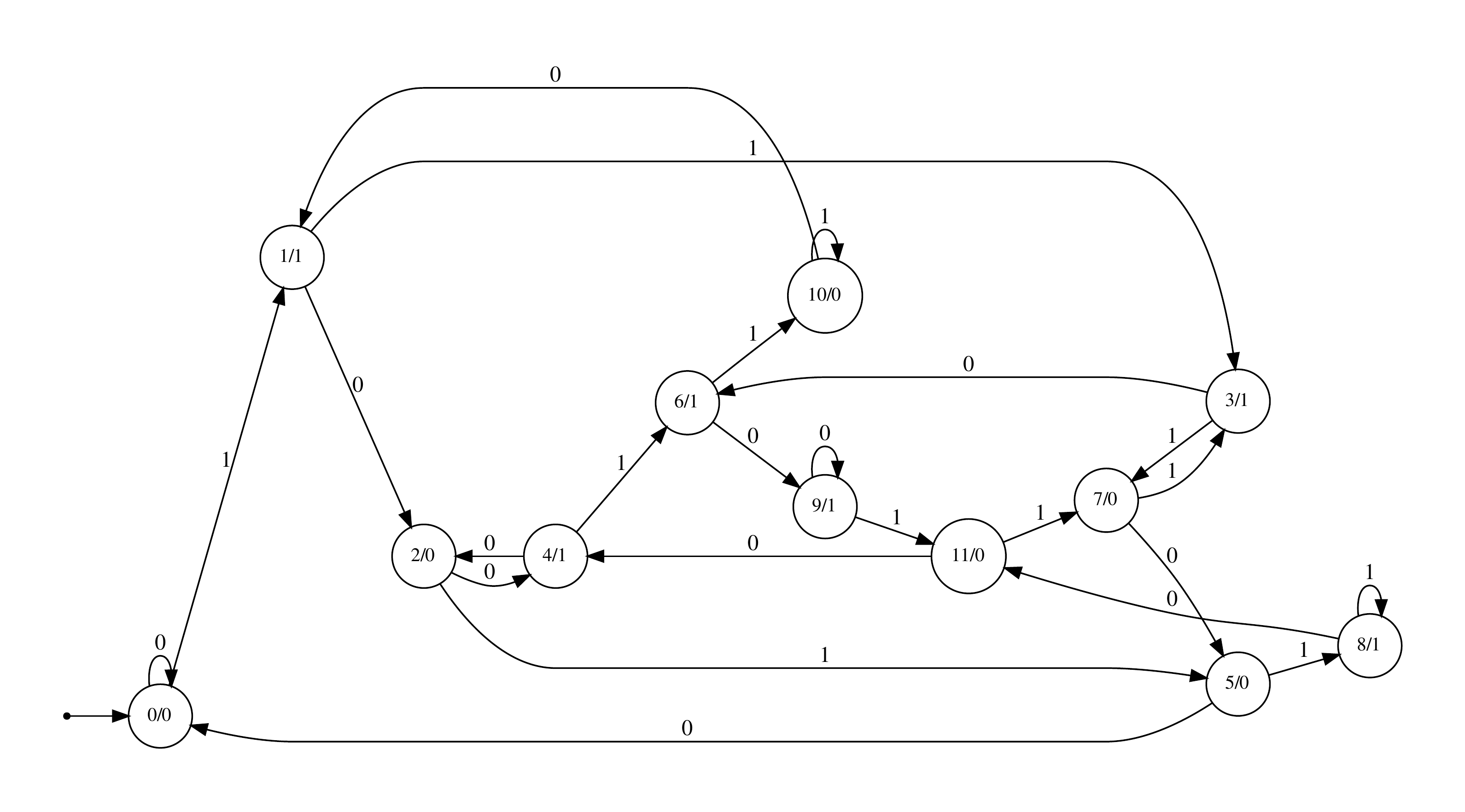}
    \end{center}
    \caption{$2$-automaton for $g(n)$.}
    \label{shev}
\end{figure}

Shevelev studied the ``record-setting'' values of
$v(n)$; these are the values $v(n)$ for which
$v(n) > v(i)$ for all $i < n$.  Shevelev defined the sequence of record-setting values as $(a(n))_{n \geq 1}$ and
the position in which they occur in $v(n)$
as $(l(n))_{n \geq 1}$.  Thus $v(l(n)) = a(n)$
for $n \geq 1$.  The first few
record-setting values are given in Table~\ref{ttwo}.
Sequence $(a(n))_{n \geq 1}$ is \seqnum{A268866} in the
OEIS.
\begin{table}[htb]
\begin{center}
\begin{tabular}{c|cc}
$n$ & $a(n)$ & $l(n)$ \\
\hline
0 & 2 & 0 \\
1 & 3 & 3 \\
2 & 22 & 10 \\
3 & 38 & 58\\
4 & 342 & 170 \\
5 & 598 & 938
\end{tabular}
\end{center}
\caption{Record-setters for $v(n)$.}
\label{ttwo}
\end{table}

The pairs $(a(n),l(n))$ are $2$-synchronized.
We can demonstrate this by creating an automaton
accepting the pairs $(a(n),l(n))_2$ in parallel.
\begin{verbatim}
def agree "Ai (i<k) => SH[i]=T[n+i]":
def vseq "$agree(k,n) & ~$agree(k+1,n)":
def recordv "$vseq(k,n) & Aj,r (r<n & $vseq(j,r)) => (j<k)": 
\end{verbatim}

The resulting automaton {\tt recordv.txt} is displayed in Figure~\ref{rec}.   
\begin{figure}[H]
\begin{center}
    \includegraphics[width=6in]{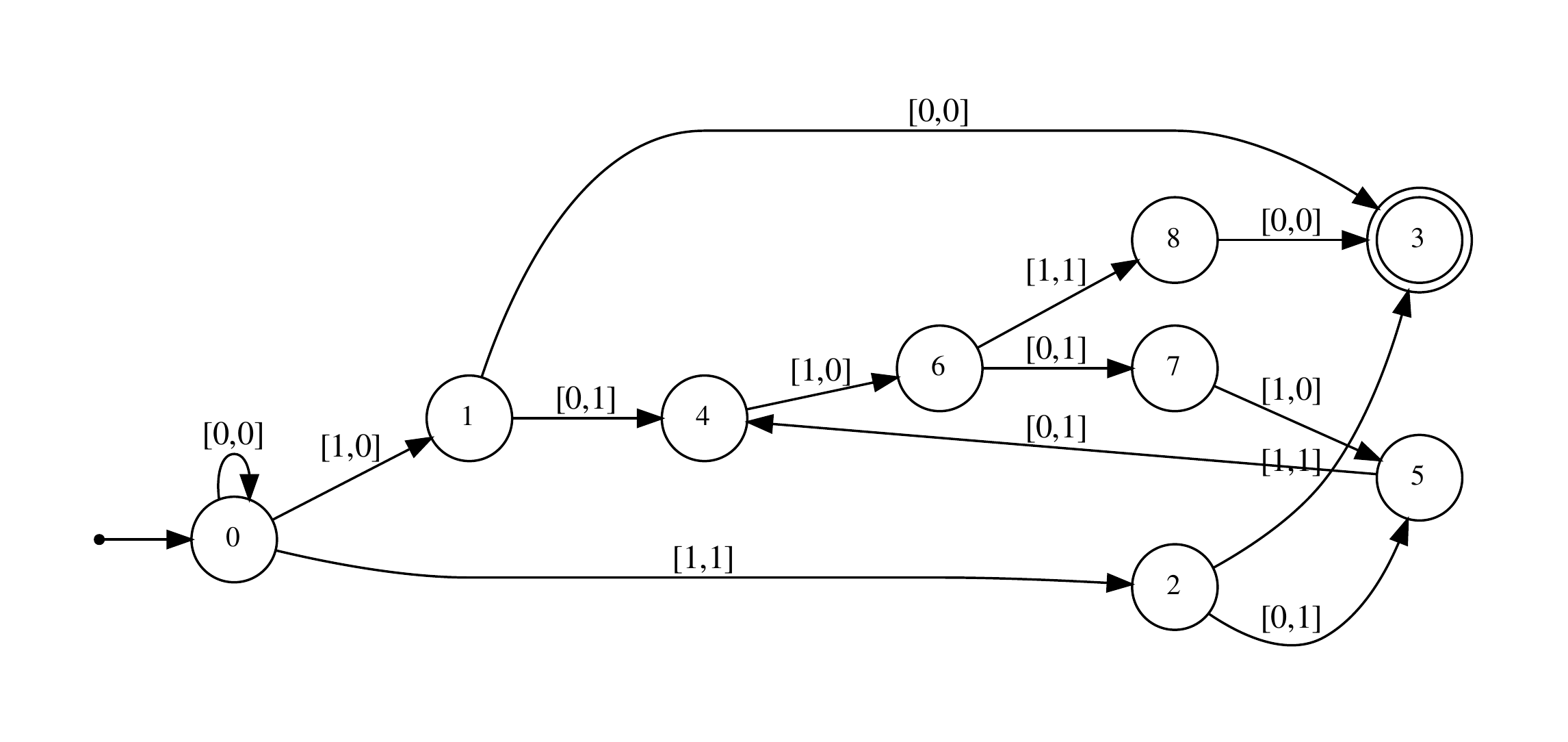}
    \end{center}
    \caption{$2$-automaton {\tt recordv.txt}.}
    \label{rec}
\end{figure}
We can now examine the acceptance paths in this automaton.  They are of four types:
\begin{align*}
& [1,0][0,0] \\
& [1,1][1,1] \\
& [1,0][0,1]([1,0][0,1][1,0][0,1])^i[1,0][1,1][0,0] \\
& [1,1][0,1][0,1][1,0]([0,1][1,0][0,1][1,0])^i[1,1][0,0] 
\end{align*}
for $i \geq 0$, representing the numbers
\begin{align*}
& (2,0) \\
& (3,3) \\
& (2^{4i+6}+2)/3, (2^{4i+5}-2)/3) \\
& ((4+7\cdot 2^{4i+5})/6, (11\cdot 2^{4i+5}-4)/6).
\end{align*}

Set $i = (n-2)/2$ for $n \geq 2$ even,
and $i = (n-3)/2$ for $n \geq 3$ odd.
Then from above
we have 
\begin{align*}
\ell(0) &= 0 \\
\ell(1) &= 3 \\
\ell(n) &= (2^{2n+1}-2)/3 \text{ for $n\geq 2$
even, and} \\
\ell(n) &= (11\cdot 2^{2n-1}-4)/6 \text{ for
$n \geq 3$ odd.}
\end{align*}  
Similarly,
\begin{align*}
a(0) &= 2 \\
a(1) &= 3 \\ 
a(n) &= (2^{2n+2}+2)/3 \text{ for $n \geq2$
even, and} \\
a(n) &= (7\cdot 2^{2n-1} + 4)/6 \text{ 
for $n \geq 3$ odd.}
\end{align*}
These statements are equivalent to the statement of open problem (C) in
Shevelev's paper, which is now proved.

\subsection{The $w(n)$ sequence}
\label{wnseq}

Shevelev also studied a ``dual problem''.  Define
 $w(n)$ to be the largest integer $k$ such
that $g(i) \not= t(n + i)$ for $0 \leq i < k$, where again $t(n)$ is the Thue-Morse sequence.  He then studied
the ``record-setting'' values of
$w(n)$; these are the values
$w(n)$ such that $w(n) > w(i)$ for all
$i < n$.  Enumerate the sequence of record-setting values as $(b(n))_{n\geq 1}$ and the position where they occur
in $v(n)$ as $(m(n))_{n \geq 1}$.  The first few record-setting values are given in Table~\ref{ttab4}.
Sequence $(b(n))_{n \geq 1}$ is sequence
\seqnum{A269341} in the OEIS.
\begin{table}[htb]
\begin{center}
\begin{tabular}{c|cc}
$n$ & $b(n)$ & $m(n)$ \\
\hline
0 & 0 & 0 \\
1 & 1 & 1 \\
2 & 6 & 2 \\
3 & 10 & 14\\
4 & 86 & 42\\
5 & 150 & 234
\end{tabular}
\end{center}
\caption{Record-setters for $w(n)$.}
\label{ttab4}
\end{table}
We can find formulas for these numbers
just as we did in Section~\ref{vnseq}.

\begin{verbatim}
def disagree "Ai (i<k) => SH[i]!=T[n+i]":
def wseq "$disagree(k,n) & ~$disagree(k+1,n)":
def recordw "$wseq(k,n) & Aj,r (r<n & $wseq(j,r)) => (j<k)": 
\end{verbatim}

The resulting automaton {\tt recordw.txt} is displayed in Figure~\ref{recw}.   
\begin{figure}[H]
\begin{center}
    \includegraphics[width=6in]{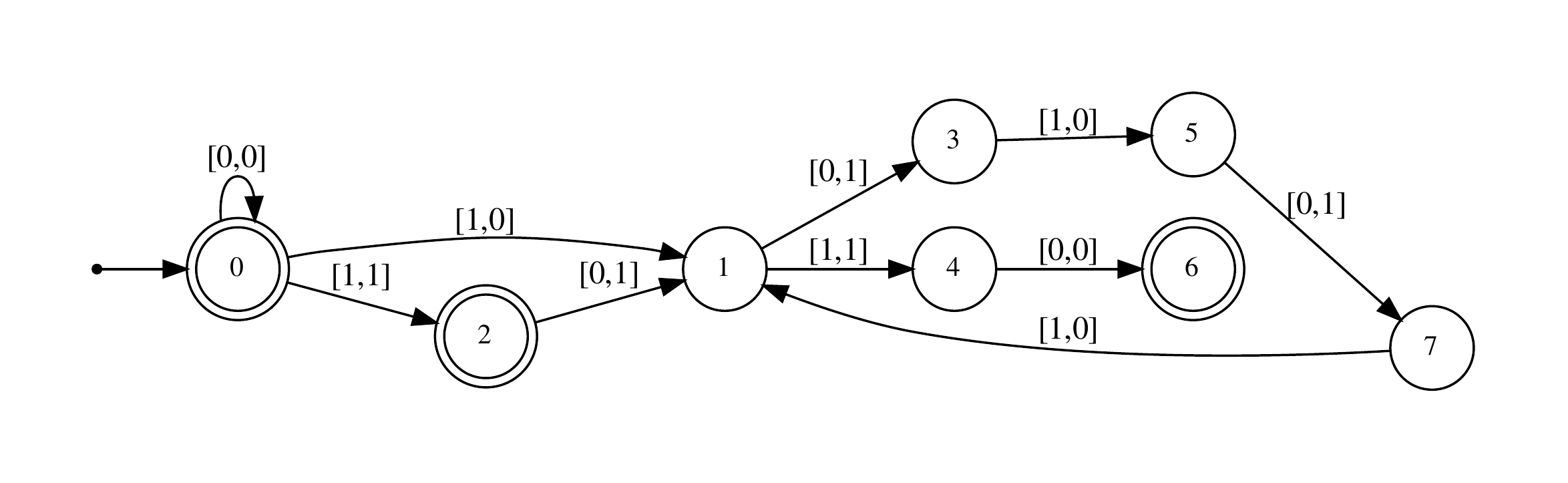}
    \end{center}
    \caption{$2$-automaton {\tt recordw.txt}.}
    \label{recw}
\end{figure}
We can now examine the acceptance paths in this automaton.  They are of four types:
\begin{align*}
& [0,0] \\
& [1,1] \\
& [1,0]([0,1][1,0][0,1][1,0])^i[1,1][0,0] \\
& [1,1][0,1]([0,1][1,0][0,1][1,0])^i[1,1][0,0] \end{align*}
for $i \geq 0$, representing (in binary) the pairs
\begin{align*}
& (0,0) \\
& (1,1) \\
& (2^{4i+4} + 2)/3, (2^{4i+3}-2)/3) \\
& ((7\cdot 2^{4i+3} +4)/6, (11\cdot 2^{4i+3}-4)/3).
\end{align*}
Set $i = (n-2)/2$ for $n \geq 2$ even,
and $i = (n-3)/2$ for $n \geq 3$ odd. 
Then
from above we have
\begin{align*}
m(0) &= 0 \\
m(1) &= 1 \\
m(n) &= (2^{2n-1} - 2)/3 \text{ for $n \geq 2$ even, and} \\
m(n) &= (11 \cdot 2^{2n-3} -4)/6 \text{ for $n\geq 3$ odd.}
\end{align*} 
Similarly,
\begin{align*}
b(0) &= 0 \\
b(1) &= 3 \\
b(n) &= (2^{2n} + 2)/3 \text{ for $n \geq 2$ even, and} \\
b(n) &= (7\cdot 2^{2n-3} +4)/6 \text{ for
$n \geq 3$ odd.}
\end{align*}
These two formulas are equivalent to 
Shevelev's open problem (D), which is
now proved.

\section{NegaFibonacci representation}
\label{nine}

Previous sections have discussed representation in 
base-$(-k)$.  In this section we discuss a more
exotic numeration system, called negaFibonacci
representation.

In the ordinary Fibonacci numeration system (also called the Zeckendorf numeration system), a natural number
$n$ is expressed as a sum
$$ n = \sum_{2 \leq i \leq t} e_i F_i,$$
where $F_0 =0$, $F_1 = 1$, $F_n = F_{n-1} + F_{n-2}$
are the Fibonacci numbers and $e_i \in \{0,1\}$.  
Such a representation is unique (up to leading
zeros) provided $e_i e_{i+1} \not=1$ for
$2 \leq i < t$.   The
binary word $e_t e_{t-1} \cdots e_2$ is called the
{\it canonical Fibonacci representation} for $n$ and
is denoted $(n)_F$. 
See, for example, \cite{Lekkerkerker:1952,Zeckendorf:1972}.
{\tt Walnut} can do calculations with numbers
represented in this numeration system; see
\cite{Mousavi&Schaeffer&Shallit:2016} for some examples.

As an example, consider the infinite (one-sided) Fibonacci
word
$$ {\bf f} = f_0 f_1 f_2 \cdots = 0100101001001 \cdots ,$$
which is the fixed point of the morphism
$0 \rightarrow 01$, $1 \rightarrow 0$ 
\cite{Berstel:1980b}.    It is known that for
$i \geq 0$, 
$f_i$ is equal to the least significant digit of the Fibonacci
representation of $i$, and hence $\bf f$
is Fibonacci-automatic.

As is well-known,
$\bf f$ is also a Sturmian characteristic word; more specifically,
if we define 
$${\bf s}_\theta = s_\theta(1) s_\theta(2) s_\theta(3) \cdots$$
where $0 < \theta < 1$ is a real irrational number and
$${\bf s}_\theta(n) = \lfloor (n+1)\theta \rfloor- 
\lfloor n \theta \rfloor,$$
then 
\begin{equation}
   f_i = s_\gamma(i+1) \label{fibb} 
\end{equation} 
for 
$\gamma = (3-\sqrt{5})/2$ and $i \geq 0$.

We now consider an extension of the one-sided word $\bf f$ to a two-sided
(bi-infinite) word.   There are essentially two different ways to do this; see \cite{Rosema&Tijdeman:2005,Du&Su:2020,Labbe&Lepsova:2021}.

One natural extension of $\bf f$ to a two-sided
or bi-infinite word is to use the equality
\eqref{fibb} for {\it all\/} integers $i$, not just for
$i \geq 0$.  Setting $g_i = s_\gamma(i+1)$ for
all $i \in \Zee$
gives us the following infinite word:
$${\bf g} = (g_i)_{i \in \Zee} = \cdots 010100100101001010.01001010010010100 \cdots .$$
Notice that, from the definition, the words $\bf f$ and $\bf g$ coincide on non-negative indices.

On the other hand, the Fibonacci numeration system can be extended
to a representation for all integers, not just
non-negative integers.  This extension, sometimes called the nega\-Fibonacci system, expresses
$n$ as a sum
$$ n = \sum_{1 \leq i \leq t} e_i (-1)^{i+1} F_i.$$
Bunder \cite{Bunder:1992} proved that
this is a unique
representation for all integers 
(up to leading zeros), provided
$e_i e_{i+1} \not=1$.  In this case the binary word
$e_t e_{t-1} \cdots e_1$ is called the {\it canonical
nega\-Fibonacci representation} for $n$ and is denoted 
$(n)_{-F}$.
Table~\ref{tab5} gives some examples of this
representation.
\begin{table}[H]
\begin{center}
\begin{tabular}{c|c} 
$n$ & $(n)_{-F}$ \\
\hline
$-8$ & 100000 \\
$-7$ &  100001 \\
$-6$ &  100100 \\
 $-5$ &  100101   \\
 $-4$ &  1010   \\
  $-3$ &  1000   \\
  $-2$ &  1001  \\
  $-1$ &  10     \\
  0 &  $\epsilon$      \\
  $1$ &  1   \\
  $2$ &  100   \\
  $3$ &  101  \\
  $4$ &  10010  \\
  $5$ &  10000  \\
  $6$ &  10001  \\
 $7$ &  10100  \\
$8$ & 10101 
 \end{tabular}
\end{center}
\caption{Examples of negaFibonacci representation.}
\label{tab5}
\end{table}

An alternative extension of $\bf f$ to a two-sided
word could then be defined as the least significant
bit of the negaFibonacci representation of $n$:  namely,
$$ {\bf h} = (h_i)_{i \in \Zee} =  \cdots 00101001001010010.010100101001001010 \cdots .$$
Here the relationship between $\bf f$ and $\bf h$ is slightly 
less apparent.

Our goal is to show the relationship between the
infinite words ${\bf f}$,  ${\bf g}$, and ${\bf h}$.   
Namely,
\begin{theorem}
   We have 
   \begin{itemize}
   \item[(a)] $g_i = g_{-3-i}$ for $i \not\in \{-2,-1\}$;
   \item[(b)] $h_i = h_{1-i}$ for $i \not\in \{ 0,1 \}$;
    \item[(c)] $g_{i-1} = h_{i+1}$ for $i \not\in \{-1,0 \}$.
   \end{itemize}
\end{theorem}

\begin{proof}
We start with the Fibonacci-synchronized automaton 
{\tt phin} given in \cite{Shallit:2021b} for $(n, \lfloor \alpha n \rfloor)$, $n \geq 0$, where $\alpha = (1+\sqrt{5})/2$.

From it we construct a Fibonacci-synchronized automaton
{\tt gn}
for $(n, \lfloor \gamma n \rfloor)$ for $n \geq 0$.
For $n \geq 1$ we have
$$ \lfloor \gamma n \rfloor = 
\lfloor (2-\alpha) n \rfloor =
\lfloor 2n - \alpha n \rfloor =
2n + \lfloor -\alpha n \rfloor =
2n - 1 - \lfloor \alpha n \rfloor,$$
where we have used the fact that
$\lfloor x \rfloor + \lfloor -x \rfloor = -1$ if
$x$ is not an integer.

We also construct a Fibonacci-synchronized
automaton {\tt gmn} for
$(n, -\lfloor -\gamma n \rfloor)$ for $n \geq 0$.

Next, we convert these two Fibonacci automata into two negaFibonacci automata, one for the non-negative values
of $n$ and one for the negative values.

We then join the two negaFibonacci automata into one called {\tt GG}
for $(n,\lfloor n \gamma \rfloor)$ that covers the entire range of $n$:  positive, negative, and zero.

From this we get another negaFibonacci automaton computing $\lfloor (n+1) \gamma \rfloor -
\lfloor n \gamma \rfloor$ for all $n \in \Zee$.

We create a negaFibonacci DFAO G2 for $g$ and
H2 for $h$, and then verify the identities (a)-(c).
Here is the {\tt Walnut} code:

\begin{verbatim}
reg shift {0,1} {0,1} "([0,0]|[0,1][1,1]*[1,0])*":
def phin "?msd_fib (s=0 & n=0) | Ex $shift(n-1,x) & s=x+1":
# Fibonacci synchronized automaton for (n, floor(n*alpha))

def gn "?msd_fib (n>=1 & Ex $phin(n,x) & s+x+1=2*n) | (n=0&s=0)":
# Fibonacci synchronized automaton for (n, floor(n*gamma))

combine GN gn:
def gmn "?msd_fib (n=0&s=0) | (n>=1) & $gn(n,s-1)":
combine GMN gmn:

rsplit GN2 [+] [+] GN:
rsplit GMN2 [-] [-] GMN:

join GG GN2[x][y] GMN2[x][y]:
# GG[x][y] is negaFibonacci synchronized DFAO for (n, floor(n*gamma))

def st "?msd_neg_fib Ex,y GG[n][x]=@1 & GG[n+1][y]=@1 & y=x+1":
# st[n] returns true if sturmian word at position n
# in negaFibonacci representation is equal to 1

def fb "?msd_neg_fib $st(n+1)":
# n'th position of the Fibonacci word, generalized to negaFibonacci
combine G2 fb:

eval parta "?msd_neg_fib Ai (G2[i]=G2[_(i+3)]) <=> (i!=_1 & i!=_2)":

reg ends1 msd_neg_fib "0*(0|10)*1":
combine H2 ends1:
eval partb "?msd_neg_fib Ai (H2[i]=H2[1-i]) <=> (i!=0 & i!=1)":

eval partc "?msd_neg_fib Ai (G2[i-1]=H2[i+1]) <=> (i!=0 & i!=_1)":
\end{verbatim}
{\tt Walnut} returns {\tt TRUE} when it evaluates
{\tt parta, partb, partc}.  This completes the proof.
\end{proof}

Here is another way to see that 
$\bf g$ and $\bf h$ are generalizations
of $\bf f$:
\begin{theorem}
    A word is a factor of $\bf f$ if and only if it is a factor of $\bf g$ if and only iff it is a factor of $\bf h$.
\end{theorem}
\begin{proof}
Recalling that ${\bf f}[0..\infty) =
{\bf g}[0..\infty)$, we can use
{\tt Walnut} to verify these claims:
\begin{verbatim}
def gfactoreq "?msd_neg_fib At (t>=0 & t<n) => G2[i+t]=G2[j+t]":
def ghfactoreq "?msd_neg_fib At (t>=0 & t<n) => H2[i+t]=G2[j+t]":
eval test1 "?msd_neg_fib Ai,n Ej (j>=0) & $gfactoreq(i,j,n)":
eval test2 "?msd_neg_fib Ai,n Ej (j>=0) & $ghfactoreq(i,j,n)":
\end{verbatim}
\noindent and {\tt Walnut} returns {\tt TRUE} for both.
\end{proof}

We now turn to an application of these
ideas.   We say that a finite word $w$ is a (weak) quasiperiod of an infinite word $\bf x$ of $w$ covers $\bf x$ by translates, where we allow $w$ to partially ``fall off'' the left edge of $\bf x$.
Lev\'e and Richomme
\cite{Leve&Richomme:2004} determined the (weak) quasiperiods of the (one-sided) Fibonacci word $\bf f$. We can find a simpler but equivalent characterization of the (weak) quasiperiods using the bi-infinite Fibonacci word $\bf g$ as follows:
\begin{theorem}
Suppose $n \geq 1$ and $F_{j-1} \leq n < F_j$ for some 
natural number $j$.
The length-$n$ word $x$ is a quasiperiod of the bi-infinite Fibonacci word $\bf g$ if and only if there exists an index
$i$, $0 \leq i \leq F_j - n - 2$, such that
$x = {\bf g}[i..i+n-1]$.  
\end{theorem}

\begin{proof}
We first check whether a number in negaFibonacci representation is Fibonacci number and 
assert that a number is the smallest Fibonacci number that is larger than $n$ as follows. 
\begin{align*}
    \isfib(x) &:= \exists i\ x = F_i \\ 
    \minfib(n,s) &:= \isfib(s) \land s>n \land \forall k\ (\isfib(k) \land k>n) \implies k \geq s
\end{align*}
The statements translate into \texttt{Walnut} as follows. 
To generate the regular expression for $\isfib$, 
we utilize the fact that the terms of the negaFibonacci system alternate in signs and 
$F_n = F_{n-1} + F_{n-3} + \cdots + F_1$
if $n$ is odd.  
\begin{verbatim}
reg isfib msd_neg_fib "(0*1(00)*|0*(10)*1)":
def minfib "?msd_neg_fib $isfib(s) & s>n & Ak ($isfib(k) & k>n) => k>=s":
\end{verbatim}
We then use the following to assert quasiperiodicity in the bi-infinite Fibonacci word. 
\begin{align*}
    \gfactoreq(i,j,n) &:= \forall t\ (t \geq 0 \land t < n) \implies {\bf g}[i+t] = {\bf g} [j+t] \\
    \quasibifib(k, n) &:= \forall i\ \exists j\ (i<j+n) \land (j \leq i) \land \gfactoreq(k,j,n)
\end{align*}
We express the full theorem as follows.
\begin{align*}
    \forall k,n\ &n\geq 1 \implies \\
    &(\quasibifib(k,n) \iff 
    \exists i, s\ \minfib(n,s) \land i \geq 0 \land i \leq s-n-2 \land \gfactoreq(k,i,n))
\end{align*}
Translating into \texttt{Walnut}, we have the following. 
\begin{verbatim}
def gfactoreq "?msd_neg_fib At (t>=0 & t<n) => G2[i+t]=G2[j+t]":
def quasibifib "?msd_neg_fib Ai Ej i<j+n & j<=i & $gfactoreq(k,j,n)":
eval thm7bifib "?msd_neg_fib Ak,n (n>=1) => ($quasibifib(k,n) 
   <=> (Ei,s $minfib(n,s) & i>=0 & i<=s-n-2 & $gfactoreq(k,i,n)))":
# 30 ms
# returns TRUE
\end{verbatim}
This completes the proof.
\end{proof}

\section{Final remarks}
\label{ten}

{\tt Walnut} can also be used to prove (or reprove) other results in Shevelev's paper \cite{Shevelev:2017}.
However, these additional problems do not deal with negative bases, so they are not the focus of this paper.  
We simply mention briefly that his
Open Problem A, also mentioned in 
\cite{Shevelev:2012}, has recently been completely solved
in \cite{Meleshko&Ochem&Shallit&Shan:2022}.

Furthermore, his Theorem 6, dealing with the critical exponent of the sequence counting (modulo $2$) the number of runs of $1$'s in the (ordinary)
binary representation of $n$, can be easily proved
and even improved as follows:
\begin{theorem}
Define $r(n)$ to be the parity of the number of runs of $1$'s in the binary representation of $n$.  Then the infinite word ${\bf r} = (r(n))_{n \geq 0}$ has
no $(4+\epsilon)$-powers, and this bound is optimal.
\end{theorem}
\begin{proof}
   It is easy to see that ${\bf r}[1..4] = 1111$, so clearly $\bf r$ has
   4th powers.   To show $\bf r$ has
   no $(4+\epsilon)$-powers, we use
   {\tt Walnut}:
\begin{verbatim}
reg runs msd_2 "0*11*(00*11*00*11*)*0*":
# number of runs of 1's in binary expansion of n is odd

combine RU runs:
# turn it into a DFAO

eval has4e "Ei,n (n>=1) & At (t<=3*n) => RU[i+t]=RU[i+t+n]":
# assert it has (4+epsilon)-powers
# Walnut returns FALSE
\end{verbatim}
\end{proof}
The sequence $\bf r$ is
sequence \seqnum{A268411} in the OEIS.

\section*{Acknowledgments}

We are very grateful to Arseny Shur for his suggestion that we
try to
prove Theorem~\ref{shur} with
{\tt Walnut}.  Some of the calculations with
{\tt Walnut} were done on the CrySP RIPPLE Facility at the University of Waterloo. Thanks to Ian Goldberg for allowing us to run computations on this machine.


\begin{thebibliography}{10}

\bibitem{Allouche&Shallit:2003}
J.-P. Allouche and J.~Shallit.
\newblock {\em Automatic Sequences: Theory, Applications, Generalizations}.
\newblock Cambridge University Press, 2003.

\bibitem{Berstel:1980b}
J.~Berstel.
\newblock Mots de {Fibonacci}.
\newblock {\em {S\'eminaire} d'Informatique {Th\'eorique}, LITP} {\bf 6-7}
  (1980--81), 57--78.

\bibitem{Bunder:1992}
M.~W. Bunder.
\newblock Zeckendorf representations using negative {Fibonacci} numbers.
\newblock {\em Fibonacci Quart.} {\bf 30} (1992), 111--115.

\bibitem{Du&Su:2020}
J.~Du and X.~Su.
\newblock On the existence of solutions for {Frenkel-Kontorova} models on
  quasi-crystals.
\newblock Arxiv preprint 2012.15594 [math.DS]. Available at
  \url{https://arxiv.org/abs/2012.15594}, 2020.

\bibitem{Grunwald:1885}
V.~{Gr\"unwald}.
\newblock Intorno all'aritmetica dei sistemi numerici a base negativa con
  particolare riguardo al sistema numerico a base negativo-decimale per lo
  studio delle sue analogie coll'aritmetica (decimale).
\newblock {\em Giornale di Matematiche di Battaglini} {\bf 23} (1885),
  203--221.
\newblock Errata, p. 367.

\bibitem{Knuth:1998}
D.~E. Knuth.
\newblock {\em The Art of Computer Programming}, Vol. 3: Seminumerical
  Algorithms.
\newblock Addison-Wesley, 3rd edition, 1998.

\bibitem{Labbe&Lepsova:2021}
S.~{Labb\'e} and J.~{Lep\v{s}ov\'a}.
\newblock A numeration system for {Fibonacci-like} {Wang} shifts.
\newblock In T.~Lecroq and S.~Puzynina, editors, {\em WORDS 2021}, Vol. 12847
  of {\em Lecture Notes in Computer Science}, pp.  104--116. Springer-Verlag,
  2021.

\bibitem{Lekkerkerker:1952}
C.~G. Lekkerkerker.
\newblock Voorstelling van natuurlijke getallen door een som van getallen van
  {Fibonacci}.
\newblock {\em Simon Stevin} {\bf 29} (1952), 190--195.

\bibitem{Leve&Richomme:2004}
F.~{Lev\'e} and G.~Richomme.
\newblock Quasiperiodic infinite words: some answers.
\newblock {\em Bull. European Assoc. Theor. Comput. Sci.} , No.\ 84, (2004),
  128--138.

\bibitem{Meleshko&Ochem&Shallit&Shan:2022}
J.~Meleshko, P.~Ochem, J.~Shallit, and S.~L. Shan.
\newblock Pseudoperiodic words and a question of {Shevelev}.
\newblock Arxiv preprint arXiv:2207.10171 [math.CO], available at
  \url{https://arxiv.org/abs/2207.10171}, 2022.

\bibitem{Mousavi:2016}
H.~Mousavi.
\newblock Automatic theorem proving in {{\tt Walnut}}.
\newblock Arxiv preprint arXiv:1603.06017 [cs.FL], available at
  \url{http://arxiv.org/abs/1603.06017}, 2016.

\bibitem{Mousavi&Schaeffer&Shallit:2016}
H.~Mousavi, L.~Schaeffer, and J.~Shallit.
\newblock Decision algorithms for {Fibonacci}-automatic words, {I:} basic
  results.
\newblock {\em RAIRO Inform. Th\'eor. App.} {\bf 50} (2016), 39--66.

\bibitem{Rosema&Tijdeman:2005}
S.~W. Rosema and R.~Tijdeman.
\newblock The {Tribonacci} substitution.
\newblock {\em Electronic J. Combinatorics} {\bf 5}(3) (2005), \#A13
  (electronic).

\bibitem{Shallit:2021b}
J.~Shallit.
\newblock Synchronized sequences.
\newblock In T.~Lecroq and S.~Puzynina, editors, {\em WORDS 2021}, Vol. 12847
  of {\em Lecture Notes in Computer Science}, pp.  1--19. Springer-Verlag,
  2021.

\bibitem{Shallit:2022}
J.~Shallit.
\newblock {\em The Logical Approach to Automatic Sequences: Exploring
  Combinatorics on Words with {\tt Walnut}}.
\newblock Cambridge University Press, 2022.
\newblock In press.

\bibitem{Shallit&Shur:2019}
J.~Shallit and A.~Shur.
\newblock Subword complexity and power avoidance.
\newblock {\em Theoret. Comput. Sci.} {\bf 792} (2019), 96--116.

\bibitem{Shevelev:2012}
V.~Shevelev.
\newblock Equations of the form $t(x+a)=t(x)$ and $t(x+a) = 1-t(x)$ for
  {Thue-Morse} sequence.
\newblock Arxiv preprint arXiv:0907.0880 [math.NT], available at
  \url{https://arxiv.org/abs/0907.0880}, 2012.

\bibitem{Shevelev:2017}
V.~Shevelev.
\newblock Two analogs of the the {Thue-Morse} sequence.
\newblock Arxiv preprint arXiv:1603.04434. Available at
  \url{https://arxiv.org/abs/1603.04434}, 2017.

\bibitem{Shur:2000}
A.~M. Shur.
\newblock The structure of the set of cube-free $\mathbb{Z}$-words in a
  two-letter alphabet ({Russian}).
\newblock {\em Izv. Ross. Akad. Nauk Ser. Mat.} {\bf 64} (2000), 201--224.
\newblock English translation in \emph{Izv. Math.} \textbf{64} (2000),
  847--871.

\bibitem{Sloane:2022}
N.~J.~A. Sloane et~al.
\newblock The on-line encyclopedia of integer sequences, 2022.
\newblock Available at \url{https://oeis.org}.

\bibitem{Zeckendorf:1972}
E.~Zeckendorf.
\newblock {Repr\'esentation} des nombres naturels par une somme de nombres de
  {Fibonacci} ou de nombres de {Lucas}.
\newblock {\em Bull. Soc. Roy. {Li\`ege}} {\bf 41} (1972), 179--182.

\end{thebibliography}
\end{document}